\documentclass[11pt]{article}
\usepackage{algorithm} 
\usepackage{algorithmic}  
\usepackage[algo2e]{algorithm2e} 
\RestyleAlgo{ruled}
\usepackage{amsthm,amsmath,amssymb}
\usepackage{multirow}
\usepackage[pdftex]{graphicx}
\usepackage{subfigure}
\usepackage{makecell}
\usepackage{booktabs}
\usepackage{array}
\usepackage{fullpage}
\usepackage{url} 
\usepackage{bm}
\usepackage{smile}
\usepackage{mathtools}
\usepackage{wrapfig}
\usepackage{lipsum}
\usepackage{mathrsfs}
\usepackage{dsfont}
\usepackage{titling}
\usepackage{epstopdf}
\usepackage{bbm}
\usepackage{relsize}
\usepackage{smile}
\usepackage{natbib}

\usepackage{multirow}
\usepackage{color}
\usepackage{mathtools}
\usepackage{caption}

\usepackage{multirow}

\usepackage[usenames,dvipsnames,svgnames,table]{xcolor}
\usepackage[colorlinks=true,
            linkcolor=blue,
            urlcolor=blue,
            citecolor=blue]{hyperref}

\usepackage{comment}
\usepackage{tikz}
\tikzset{
	declare function={
		normcdf(\x,\m,\s)=1/(1 + exp(-0.07056*((\x-\m)/\s)^3 - 1.5976*(\x-\m)/\s));
		zeroone(\x)= (\x<=-2) * (\x*\x + 6*\x + 8)   +
		and(\x>-2, \x<=1) * (2 - \x - \x*\x)     +
		and(\x>1,  \x<=2) * (6 - 8*\x + 2*\x*\x) +
		(\x>2) * (-10 + 6*\x - \x*\x);
	}
}

\usepackage{mathtools}

\usepackage{enumitem}

	\makeatletter
	\newcommand{\mypm}{\mathbin{\mathpalette\@mypm\relax}}
	\newcommand{\@mypm}[2]{\ooalign{%
			\raisebox{.1\height}{$#1+$}\cr
			\smash{\raisebox{-.6\height}{$#1-$}}\cr}}
	\makeatother

\def\beq{\begin{equation}\begin{aligned}[b]}
	
\def\eeq{\end{aligned}\end{equation}}

\begin{document}

\title{Treatment Effect Estimation with Unobserved and Heterogeneous Confounding Variables}

\author{Kevin Jiang\thanks{Department of Statistics and Data Science, Cornell University, Ithaca, NY 14850, USA; e-mail: \texttt{kcj42@cornell.edu}}~~~~~Yang Ning\thanks{Department of Statistics and Data Science, Cornell University, Ithaca, NY 14850, USA; e-mail: \texttt{yn265@cornell.edu}.}
}
\date{\today}

\maketitle

\vspace{-0.3in}

\begin{abstract}
The estimation of the treatment effect is often biased in the presence of unobserved confounding variables which are commonly referred to as hidden variables. Although a few methods have been recently proposed to handle the effect of hidden variables, these methods often overlook the possibility of any interaction between the observed treatment variable and the unobserved covariates. In this work, we address this shortcoming by studying a multivariate response regression problem with both unobserved and heterogeneous confounding variables of the form $\bY=\bA^T\bX+\bB^T\bZ+ \sum_{j=1}^{p} \bC^T_j X_j \bZ + \bE$, where $\bY \in \RR^m$ are $m$-dimensional response variables, $\bX \in \RR^p$ are observed covariates (including the treatment variable), $\bZ \in \RR^K$ are $K$-dimensional unobserved confounders, and $\bE \in \RR^m$ is the random noise. Allowing for the interaction between $X_j$ and $\bZ$ induces the heterogeneous confounding effect. Our goal is to estimate the unknown matrix $\bA$, the direct effect of the observed covariates or the treatment on the responses. To this end, we propose a new debiased estimation approach via SVD to remove the effect of unobserved confounding variables. The rate of convergence of the estimator is established under  both the homoscedastic and heteroscedastic noises. We also present several simulation experiments and a real-world data application to substantiate our findings.  
    
\end{abstract}

\section{Introduction}
Treatment effect estimation in the presence of unobserved confounding variables is a very challenging problem, arising in many areas including statistics, biology, computer science and economics. With some additional domain knowledge, such as the existence of instrumental variables or negative controls, the effect of unobserved confounding variables can be removed, leading to consistent estimation of the treatment effect; see  \cite{wooldridge2015introductory,lipsitch2010negative} for a review. However, if such information is unavailable, how to correct for the bias due to the unobserved confounding variables is largely unexplored. 

In this context, a class of methods known as surrogate variable analysis (SVA) has been proposed to account for the hidden variables (e.g., batch effect) in the analysis of genomics data (e.g. \cite{alter2000singular}, \cite{leek2007capturing}, \cite{sun2012multiple}). These methods relax the assumptions commonly used in the literature on instrumental variables or negative controls, but still require substantial apriori knowledge of the data or impose a strict structure on the underlying model. For example, \cite{gagnon2012using} required knowing a null set of features in the exposure variables, \cite{mckennan2019accounting} required row-wise sparsity of the non-null features, and \cite{wang2017confounder} imposed a linear causal relationship of the observed variables and the hidden variables. More recently, \cite{bing2022adaptive} extended the methods to deal with the model with high-dimensional features. However, all of these methods ignore the possibility of any interaction between the observed treatment variable and the hidden variables, which appears frequently in the presence of heterogeneous treatment effect (e.g., the treatment effect may vary according to the value of the confounding variables). Failing to account for this structure may lead to model specification, so that the validity and interpretability of statistical findings may be severely limited. 

In this paper, we consider the following model 
\begin{align}\label{eq:first_model}
    \bY=\bA^T\bX+\bB^T\bZ+ \sum_{j=1}^{p} \bC^T_j X_j \bZ + \bE, 
\end{align}
where $\bY \in \RR^m$ are the response variables, $\bX \in \RR^p$ are the observed covariates (including the treatment variable and observed confounders), $\bZ \in \RR^K$ are the unobserved covariates or hidden variables, and $\bE \in \RR^m$ are the random errors. The matrices $\bA, \bB, \bC_1,...,\bC_p$ are unknown parameters. The model  (\ref{eq:first_model}) allows the hidden variables to interact through the $\bC^T_j X_j \bZ$ terms for $j=1,\dots,p$ in a multiplicative manner. As a special case of (\ref{eq:first_model}), consider $p=1$ and take $X_1\in\{0,1\}$ to be a binary treatment variable. Then the conditional average treatment effect (CATE) can be shown to be $\bA^T+\bC^T_j \bZ$, which depends on the value of unobserved confounding variable $\bZ$. Thus, model  (\ref{eq:first_model}) provides a parsimonious way to account for the heterogeneous treatment effect. 

Given $n$ i.i.d samples $(\bY_i,\bX_i)_{i=1}^n$, our goal is to estimate the unknown matrix $\bA\in\RR^{p\times m}$, the association between $\bX$ and $\bY$ in the presence of hidden variables, which can be also interpreted as the direct effect of the treatment on the response in a causal inference framework \citep{bing2022adaptive}. However, in general, $\bA$ is not identifiable by only observing $(\bY, \bX)$, as $\bZ$ can be correlated with $\bX$ in an arbitrary way. To address this problem, we construct a parameter $\bTheta$ that approximates $\bA$ by teasing out the effect induced by the hidden variables. In particular, we propose a debiased estimator by projecting the response variables to an appropriate singular vector space. Theoretically, we characterize the stochastic error and approximation error of our debiased estimator under both the homoscedastic and the more general heteroscedastic and correlated noises.




The paper is organized as follows. We first give a detailed estimation algorithm in the homoscedastic setting in Section \ref{sec:method}. Section \ref{sec:theoretical} presents our main theoretical result concerning the convergence rate of our debiased estimator. We then extend the method to the heteroscedastic setting in Section \ref{sec:heteroscedastic} by proposing a modified version of our algorithm and giving an adjusted convergence rate. Finally, Sections \ref{sec:simulation} and \ref{sec:dataapplication} give simulation results and a real world data application to high-throughput microarray data.





\subsection{Notation}
For any set $\cS$, we write $|\cS|$ for its cardinality. For any vector $\bv\in \RR^d$, we define its $\ell_q$ norm as $\|\bv\|_q = (\sum_{j=1}^d |\bv_j|^q)^{1/q}$ for some real number $q\ge 0$. For any matrix $\bM \in \RR^{d_1 \times d_2}$, we denote $\|\bM\|_{op}$ and $\|\bM\|_{F}$ as the operator and Frobenius norm, respectively. For any square matrix $\bM$, we also write $\lambda_k(\bM)$ as its $k$th largest eigenvalue. For any two sequences $a_n$ and $b_n$, we write $a_n \lesssim b_n$ (or $a_n=\mathcal{O}(b_n)$) if there exists some positive constant $C$ such that $a_n \le Cb_n$ for any $n$. We let $a_n \asymp b_n$ stand for $a_n\lesssim b_n$ and $b_n \lesssim a_n$. Denote $a\vee b=\max (a,b)$ and $a\wedge b=\min(a,b)$. 

\section{Debiased Estimator via SVD}\label{sec:method}
Recall that in model (\ref{eq:first_model}), $\bY \in \RR^m$ are the response variables, $\bX \in \RR^p$ are the observed covariates and $\bZ \in \RR^K$ are the hidden variables, where the number of hidden variables $K$ is unknown and is assumed to be much less than $m$. In addition, we assume the random noise $\bE$ is independent of $\bX$ and $\bZ$. In this section, we focus on the setting with homoscedastic errors, i.e., $\Cov(\bE) = \sigma^2 \bI_m$, where $\bI_m$ is an identity matrix.  The extension to heterogeneous and correlated errors is studied in Section \ref{sec:heteroscedastic}. 

To motivate the proposed method, write $\bW=\bZ-\bpsi^T \bX$, where $\bpsi := \{ \EE(\bX \bX^T)\}^{-1} \EE(\bX\bZ^T) \in \RR^{p \times K}$ is obtained by the $\mathcal{L}_2$ projection of $\bZ$ onto $\bX$. Note that we do not require $\bZ$ and $\bX$ to be linearly dependent when defining $\bW$. Using this decomposition of $\bZ$, we may rewrite (\ref{eq:first_model}) as 
\begin{align}\label{model}
\bY &= \bA^T \bX + (\bB^T \bW + \bB^T \bpsi^T \bX) + \sum_{j=1}^{p} \bC_j^T X_j (\bW + \bpsi^T \bX) + \bE \notag \\
&= \sum_{j=1}^{p} ( \bA_{(\cdot,j)}^T + \bB^T \bpsi_{(\cdot,j)}^T ) X_j + \sum_{j=1}^{p} \bC_j^T X_j \cdot \sum_{k=1}^{p} \bpsi_{(\cdot,k)}^T X_k + \bB^T \bW + \sum_{j=1}^{p} \bC_j^T X_j \bW + \bE \notag \\
&= \sum_{j=1}^{p} \bL_{1j}^T X_j + \sum_{1\leq j\leq k\leq p} \bL_{2,jk}^T X_j X_k + \bepsilon,
\end{align}
where $\bL_{1j}^T =  \bA_{(\cdot,j)}^T + \bB^T \bpsi_{(\cdot,j)}^T\in\RR^m$, $\bL_{2,jk}^T = \bC_j^T \bpsi_{(\cdot,k)}^T+\bC_k^T \bpsi_{(\cdot,j)}^T\in \RR^m$ for $j\neq k$ and $\bL_{2,jk}^T = \bC_j^T \bpsi_{(\cdot,k)}^T\in \RR^m$ for $j=k$, and $\bepsilon = \bB^T \bW + \sum_{j=1}^{p} \bC_j^T X_j \bW + \bE$. Thus, given $n$ i.i.d copies of $(\bX, \bY)$, we can estimate the coefficient matrices $\bL_{1j}, \bL_{2,jk}$ from a linear regression with all the linear and pairwise interactions among $\bX$. 

The second step is to estimate the covariance matrix of the residuals $\bepsilon$, which has the following structures
\begin{align}\label{covariance residual}
\EE(\bepsilon\bepsilon^T|\bX)&=\EE[(\bB^T \bW + \bE)(\bB^T \bW + \bE)^T|\bX] + \sum_{1\leq j,k\leq p} \EE(\bC_j^T \bW\bW^T\bC_k|\bX) X_jX_k \notag \\
&~~~+\sum_{j=1}^{p} \EE(\bC_j^T \bW(\bB^T \bW + \bE)^T|\bX) X_j + \sum_{j=1}^{p} \EE((\bB^T \bW + \bE)\bW^T\bC_j|\bX) X_j \notag \\
&=\bphi_{\bB} + \sum_{1\leq j\leq k\leq p} \bphi_{(\bC_j,\bC_k)} X_jX_k + \sum_{j=1}^{p} \bphi_{(\bB, \bC_j)} X_j
\end{align}
where we use the fact that $\bE$ is independent of $\bX, \bZ$, and $\bphi_{\bB} = \bB^T \Cov(\bW|\bX) \bB + \Cov(\bE)$, $\bphi_{(\bC_j,\bC_k)} = \bC_j^T \Cov(\bW|\bX) \bC_k+\bC_k^T \Cov(\bW|\bX) \bC_j$ for $j\neq k$ and $\bphi_{(\bC_j,\bC_k)} = \bC_j^T \Cov(\bW|\bX) \bC_k$ for $j=k$, and $\bphi_{(\bB,\bC_j)}=\bC_j^T \Cov(\bW|\bX) \bB+\bB^T\Cov(\bW|\bX)\bC_j$. 

If $\bSigma_W:=\Cov(\bW|\bX)$ does not depend on $\bX$, we can regress the estimated covariance matrix of the residuals on $\bX$ to estimate $(\bphi_{\bB}, \bphi_{(\bC_j,\bC_k)}, \bphi_{(\bB,\bC_j)})$ for $1\leq  j,k\leq p$. For notational simplicity, we write $\bphi_{(\bC_j,\bC_j)}$ as $\bphi_{\bC_j}$. Suppose $\bphi_{\bB}$ and $\bphi_{\bC_j}$ for $1\leq j\leq p$ are known (or well estimated via least square estimation). Under some conditions detailed in Section \ref{sec:theoretical}, we can recover the right singular space of $\bB$ and $\bC_j$ for $1\leq j\leq p$. With some simple algebra, this gives us the projection matrix $\bP_{\bD}=\bD^T(\bD\bD^T)^{-1}\bD$, where $\bD^T := (\bB^T, \bC_1^T...,\bC_p^T)\in\RR^{m\times (p+1)K}$. Define $\bP_{\bD}^{\perp}=\bI_m-\bP_{\bD}$. 

The key step in our method is to project the $m$-dimensional response $\bY$ to the orthogonal complement of the right singular space of $\bD$. Multiplying $\bP_{\bD}^{\perp}$ on both sides of equation  (\ref{eq:first_model}), we get 
\begin{align}
\bP_{\bD}^{\perp}\bY&=\bP_{\bD}^{\perp}\bA^T\bX+\bP_{\bD}^{\perp}\bB^T\bZ+ \sum_{j=1}^{p} \bP_{\bD}^{\perp}\bC^T_j X_j \bZ + \bP_{\bD}^{\perp}\bE  \nonumber\\
&=\bP_{\bD}^{\perp}\bA^T\bX+ \bP_{\bD}^{\perp}\bE,\label{eq_projection}
\end{align}
where we use the fact that $\bP_{\bD}^{\perp}\bB^T=\bP_{\bD}^{\perp}\bC^T_j=0$ by the definition of the projection matrix. Thus, by leveraging the singular vector space of $\bB$ and $\bC_j$, we eliminate the effect of hidden variables $\bZ$ from the linear regression as shown in  (\ref{eq_projection}). However, in this case, we can only recover the coefficient matrix $\bTheta^T := \bP_{\bD}^{\perp} \bA^T$, which differs from the parameter of interest $\bA^T$ in general. The difference between $\bTheta$ and $\bA$ leads to an approximation error in our theoretical analysis. In particular, under the conditions in Section \ref{sec:theoretical}, we show that the approximation error is asymptotically ignorable.



For clarity, we summarize our debiased estimation procedure in Algorithm \ref{algo:homoscedastic}. Note that the algorithm requires the number of hidden variables $K$ as the input. The discussion on selection of $K$ is deferred to Section \ref{sec:k-selection}. 

\begin{algorithm}
\caption{Debiased estimator with homoscedastic noise $\Cov(\bE) = \sigma^2 \bI_m$}\label{algo:homoscedastic}

 \textbf{Require: } $n$ i.i.d data $(\bY_i,\bX_i)_{i=1}^n$, rank $K$; 
 
 1. Solve $(\hat{\bL}_{1j}, \hat{\bL}_{2,jk})_{1\leq j,k\leq p} = \arg \min \sum_{i=1}^{n}\| \bY_i - \sum_{j=1}^{p} \bL_{1j}^T X_{ji} - \sum_{1\leq j\leq k\leq p} \bL_{2,jk}^T X_{ji}X_{ki}\|_2^2$. 
 
 2. Obtain $\hat\Sigma_{\bepsilon_i} = \hat{\bepsilon}_i \hat{\bepsilon}_i^T \in \RR^{m \times m}$ where $\hat{\bepsilon}_i = \bY_i - \sum_{j=1}^{p} \hat\bL_{1j}^T X_{ji} - \sum_{1\leq j\leq k\leq p} \hat\bL_{2,jk}^T X_{ji}X_{ki}$. 
 
 3. Solve \begin{align*}
     (\hat\bphi_{\bB}, \hat\bphi_{(\bC_j,\bC_k)}, \hat\bphi_{(\bB,\bC_j)})_{1\leq j,k\leq p} = \arg \min \sum_{i=1}^{n} \|\hat\Sigma_{\bepsilon_i} - \bphi_{\bB} - \sum_{1\leq j\leq k\leq p} \bphi_{(\bC_j,\bC_k)} X_{ji}X_{ki} - \sum_{j=1}^{p} \bphi_{(\bB, \bC_j)} X_{ji}\|_F^2.
 \end{align*} 
 
 4. Compute $\hat \bU_{\bB}\in\RR^{m\times K}$ and $\hat \bU_{\bC_j}\in\RR^{m\times K}$, the first $K$ eigenvectors of $\hat{\bphi}_{\bB}$ and $\hat{\bphi}_{\bC_j}:=\hat\bphi_{(\bC_j,\bC_j)}$. Via SVD, obtain $\hat{\bU}_{\bD} \in \RR^{m \times (p+1) K}$, the first $(p+1) K$ left singular vectors of $(\hat{\bU}_{\bB},\hat{\bU}_{\bC_1},...,\hat{\bU}_{\bC_p})$. Then estimate the projection matrix $\hat{\bP}_{\bD} = \hat{\bU}_{\bD} \hat{\bU}_{\bD}^T$, and $\hat{\bP}_{\bD}^{\perp}=\bI_m-\hat{\bP}_{\bD}$. 
 
 5. Obtain the debiased estimator $\hat{\bTheta}$ by solving \begin{align}\label{eq:step 5}
     \hat{\bTheta} = \arg \min_{\Theta^T \in \RR^{m\times p}} \sum_{i=1}^{n} \| \hat{\bP}_{\bD}^{\perp} \bY_{i} - \bTheta^T \bX_i\|_2^2.
 \end{align}

\end{algorithm}

\section{Theoretical Results}\label{sec:theoretical}
In this section, we establish the rate of convergence of the proposed debiased estimator. In particular, we focus on the asymptotic regime with $p$ fixed  and $K,m,n\rightarrow\infty$. 

\begin{assumption}\label{ass_1}
Assume that $\bSigma_W:=\Cov(\bW|\bX)$ does not depend on $\bX$, the matrix $\bSigma_W$ has full rank and the noise is homoscedastic, i.e., $\Cov(\bE) = \sigma^2 \bI_m$. 
\end{assumption}

\begin{assumption}\label{ass_2}
Denote $\bar \bX=(X_1,...,X_p, X_1^2, X_1X_2,...,X_p^2)^T\in\RR^{p(p+1)/2}$. Assume that the matrix $\EE(\bar \bX\bar \bX^T)$ is positive definite. We further assume $\max_{1\leq k\leq m}\EE(\bar \bX^T\bar \bX\epsilon^2_k)$ and $\max_{1\leq j,k\leq m}\EE(\bar \bX^T\bar \bX\phi^2_{jk})$ are bounded by a constant, where $\epsilon_k$ is the $k$th entry of $\bepsilon = \bB^T \bW + \sum_{j=1}^{p} \bC_j^T X_j \bW + \bE$, and $\phi_{jk}$ is the $(j,k)$th entry of $\bphi=\bepsilon\bepsilon^T-\bphi_{\bB} - \sum_{1\leq j\leq k\leq p} \bphi_{(\bC_j,\bC_k)} X_jX_k - \sum_{j=1}^{p} \bphi_{(\bB, \bC_j)} X_j$. Finally, for $j=1,..,p$, $\EE(X_j^6)$ is finite.  
\end{assumption}

\begin{assumption}\label{ass_3}
Let $\lambda_B$ and $\lambda_{C_j}$ be the $K$th largest eigenvalues of $\bB^T \bSigma_W \bB$ and $\bC_j^T \bSigma_W \bC_j$, respectively. Assume $\lambda_B \asymp \lambda_{C_j} \asymp m$ for $j=1,..,p$. The matrix $\bD^T := (\bB^T, \bC_1^T...,\bC_p^T)\in\RR^{m\times (p+1)K}$ has rank $(p+1)K$, and the $(p+1)K$th largest eigenvalue of $\bD^T \bD$ denoted by $\lambda_D$ satisfies $\lambda_D\asymp m$.
\end{assumption}

Assumption \ref{ass_1} is already required in Section \ref{sec:method} in order to develop the proposed method. Assumption \ref{ass_2} is a standard moment condition that guarantees the desired rate for the least square type estimators in Algorithm \ref{algo:homoscedastic}. Finally, Assumption \ref{ass_3} is known as the pervasiveness assumption in the factor model literature for identification and consistent estimation of the right singular space of $\bB$ and $\bC_j$; see \cite{bai2003inferential,fan2013large}. In particular, $\lambda_B \asymp m$ holds if the smallest and largest eigenvalues of $\bSigma_W$ are bounded away from 0 and infinity by some constants, and the columns of $\bB$ are i.i.d. copies of a $K$-dimensional sub-Gaussian random vector whose covariance matrix has bounded eigenvalues. In this assumption, we also require that $\bD^T$ has full column rank (and $(p+1)K\leq m$), which is used to construct the estimated singular vectors $\hat{\bU}_{\bD}$ in  Algorithm \ref{algo:homoscedastic}. 

We are now ready to present the main result concerning the convergence rate of our estimator $\hat{\bTheta}$ to the true coefficient matrix $\bA$.

\begin{theorem}\label{thm:thm 1}
    Under Assumptions \ref{ass_1}-\ref{ass_3}, the estimator $\hat{\bTheta}$ from Algorithm \ref{algo:homoscedastic} satisfies:
    \begin{equation}
    \frac{1}{m}\| \hat{\bTheta} - \bA \|_F^2 = \mathcal{O}_p\Big(\frac{1}{n}+\frac{K}{m}\eta\Big),
    \end{equation}
where $\eta=\frac{1}{Km}\|\bD \bA^T\|_F^2$.
\end{theorem}

Note that the term $\frac{1}{m}\| \hat{\bTheta} - \bA \|_F^2$ can be interpreted as the squared error per response. This theorem shows that the error can be decomposed into two parts, $\mathcal{O}_p(\frac{1}{n})$ the typical parametric rate for estimating a finite dimensional parameter and  $\mathcal{O}_p(\frac{K}{m}\eta)$ the approximation error due to the difference between $\bTheta$ and $\bA$. To further examine the approximation error, we can show that $\|\bD \bA^T\|_F^2=\mathcal{O}_p(Km)$, when each row of $\bA$ and $\bD$ are independently generated from $N(0, \bI_m)$ (or more generally $N(0, \bSigma)$ where $\bSigma$ has a bounded operator norm). This implies $\eta=\mathcal{O}_p(1)$. In addition, if $m\gg nK$, the approximation error is asymptotically ignorable compared to the parametric rate, leading to    $\frac{1}{m}\| \hat{\bTheta} - \bA \|_F^2 = \mathcal{O}_p(\frac{1}{n})$ which is the best possible rate for estimating a finite dimensional parameter in regular parametric models. Even if $m\gg nK$ does not hold, the estimator is still consistent in terms of the squared error per response, as long as $K/m\rightarrow 0$. 

Finally, we note that when $\eta=\mathcal{O}_p(1)$, the error bound decreases with $m$, which implies that by collaborating more types of responses, the treatment effect estimation can be more accurate. This phenomenon can be viewed as the bless of dimensionality in our problem.

\section{Extension to Heteroscedastic and Correlated Noise}\label{sec:heteroscedastic}
In this section, we generalize our method to the setting with heteroscedastic and correlated errors, i.e.,  $\Cov(\bE)$ can be any positive definite matrix. Recall that $\bphi_{\bB} = \bB^T \Cov(\bW|\bX) \bB + \Cov(\bE)$. In view of Algorithm \ref{algo:homoscedastic}, when the noise is heteroscedastic and correlated, the main challenge is that the eigenspace of $\bphi_B$ corresponding to the first $K$ eigenvalues no longer coincides with the right singular space of $\bB$. Because of this, we may no longer be able to identify $\bP_D$ via the eigenspaces of the coefficient matrices in step 4 of Algorithm \ref{algo:homoscedastic}. To address this problem, we turn to a recently developed procedure called HeteroPCA \citep{zhang2022heteroskedastic} that allows us to recover the desired right singular space of $\bB$ by iteratively imputing the diagonal entries of $\bphi_B$ via the diagonals of its low-rank approximations. For completeness, we restate their procedure in Algorithm \ref{algo:heteroPCA}. The estimation of $\bA$ remains nearly identical to the homoscedastic setting, except we now estimate $\tilde{\bU}_{\bB}$ by performing HeteroPCA rather than PCA on the coefficient matrix $\hat{\bphi}_{\bB}$ obtained from step 3 in Algorithm \ref{algo:homoscedastic}. The full procedure is detailed in Algorithm \ref{algo:heteroscedastic}. This algorithm requires to specify the number of iterations $T$ in the HeteroPCA algorithm. Usually, a small $T$ such as $T = 5$ yields satisfactory results in our simulations.


To establish the rate of convergence of the debaised estimator in Algorithm \ref{algo:heteroscedastic}, we further impose the following condition.

\begin{assumption}\label{ass_4}
Let $\bU \in \RR^{m \times K}$ be the first $K$ left singular vectors of $\bB^T$, and $\lambda_1$ and $\lambda_B$ be the first and $K$th largest eigenvalues of $\bB^T \Sigma_{\bW} \bB$, respectively. Let $\{e_j\}_{j=1}^{m}$ denote the canonical basis vector of $\RR^m$. Then, $\exists \hspace{1mm} C_U > 0$ such that 
\begin{align*}
    \frac{\lambda_1}{\lambda_B} \max_{1\leq j\leq m} \|e_j^T U\|_2^2 \leq C_U.
\end{align*} 
\end{assumption}

This condition is a variation of the incoherence condition in the matrix completion literature \citep{candes2010power}, which is mainly used to recover the singular vector space $\bU$ in the HeteroPCA algorithm. This assumption controls how the singular vector space $\bU$ may coincide with the canonical basis vectors. Intuitively, if $\bU$ becomes more aligned with the canonical basis vectors, it is more difficult to separate $\bB^T \Cov(\bW|\bX) \bB$ and $\Cov(\bE)$ from $\bphi_{\bB}$, even if $\Cov(\bE)$ is just a diagonal matrix. 

We note that, unlike the original HeteroPCA algorithm \citep{zhang2022heteroskedastic} which requires the error matrix $\Cov(\bE)$ to be a diagonal matrix, we also allow $\Cov(\bE)$ to have non-zero off-diagonal entries (i.e., correlated errors). For any matrix $\bM$, let $D(\bM)$ be the matrix with diagonal entries equal to the diagonal entries of $\bM$, but with all off-diagonal entries equal to 0. Let $\Gamma(\bM) = \bM - D(\bM)$. The following theorem shows the rate of convergence of the debiased estimator in Algorithm \ref{algo:heteroscedastic}. 


\begin{theorem}\label{thm:thm 2}  
Assume that $\bSigma_W=\Cov(\bW|\bX)$ does not depend on $\bX$, the matrix $\bSigma_W$ has full rank and Assumptions \ref{ass_2}-\ref{ass_4} hold. If $\|\Gamma(\Cov(\bE))\|_F = \mathcal{O}(\lambda_B\sqrt{K})$ and $T \gg (1 \vee \log \frac{\sqrt{K} \lambda_B}{\|\Gamma(\Cov(E))\|_F})$, then the debiased estimator $\hat{\bTheta}$ from Algorithm \ref{algo:heteroscedastic} satisfies:
    \begin{equation}\label{eq_rate_2}
    \frac{1}{m}\| \hat{\bTheta} - \bA \|_F^2 = \mathcal{O}_p\Big(\frac{1}{n}+\frac{K}{m}\eta+\frac{\bar \eta}{m^2}\Big),
    \end{equation}
where $\eta=\frac{1}{Km}\|\bD \bA^T\|_F^2$ and $\bar\eta=\|\Gamma(\Cov(\bE))\|_F^2$.
\end{theorem}

Compared to the results in Theorem \ref{thm:thm 1}, the error bound with heteroscedastic and correlated noise includes an additional term $\mathcal{O}_p(\frac{\bar \eta}{m^2})$, which comes from the correlation among the noise vector $\bE$. In particular, if the correlation among $\bE$ is relatively weak with $\bar\eta\ll m^2/n$ and $\eta\ll m/(nK)$ holds as discussed after Theorem \ref{thm:thm 1}, we obtain that the squared error per response has the parametric rate $\mathcal{O}_p(\frac{1}{n})$. 

Finally, we comment that if one applies the proposed Algorithm \ref{algo:homoscedastic} tailored for the homoscedastic noise to deal with the model with heteroscedastic and correlated noise in this section, the rate of the estimator would be slower than (\ref{eq_rate_2}), as it contains an additional error related to the degree of heteroscedasticity of the noise variance. We refer to \cite{bing2022adaptive} for a similar result when the model has no interaction terms.

\begin{algorithm}
\caption{HeteroPCA($\hat{\Sigma}$, $K$, $T$)} \label{algo:heteroPCA}
\textbf{Require: } Matrix $\hat{\Sigma} \in \RR^{m \times m}$, rank $K$, number of iterations $T$;

1. Set $N_{ij}^{(0)} = \hat{\Sigma}_{ij}$ for all $i \neq j$ and $N_{ii}^{(0)}$ for all $i$.

2. For $t=0,1,...,T$:

\hspace{5mm} Compute SVD of $N^{(t)} = \sum_{i=1}^{m} \lambda_i^{(t)} u_i^{(t)} (v_i^{(t)})^T$, where $\lambda_1^{(t)} \geq \lambda_2^{(t)} \geq ... \geq 0$.

\hspace{5mm} Let $\tilde{N}^{(t)} = \sum_{i=1}^{r} \lambda_i^{(t)} u_i^{(t)} (v_i^{(t)})^T$.

\hspace{5mm} Set $N_{ij}^{(t+1)} = \hat{\Sigma}_{ij}$ for all $i \neq j$ and $N_{ii}^{(t+1)} = \tilde{N}_{ii}^{(t)}$.

3. Return $U^{(T)} = (u_1^{(T)}, ..., u_r^{(T)})$.
\end{algorithm}

\begin{algorithm}
\caption{Debiased estimator with heteroscedastic and correlated noise}\label{algo:heteroscedastic}

 \textbf{Require: } $n$ i.i.d data $(\bY_i,\bX_i)_{i=1}^n$, rank $K$, number of iterations $T$;
 
 1-3. Same as steps 1-3 of Algorithm \ref{algo:homoscedastic}. 
 
 4. Compute $\tilde{\bU}_{\bB}$ from HeteroPCA($\hat{\bphi}_{\bB}$, $K$, $T$) in Algorithm \ref{algo:heteroPCA}. Then compute $\hat \bU_{\bC_j}$, the first $K$ eigenvectors of, $\hat{\bphi}_{\bC_j}:=\hat\bphi_{(\bC_j,\bC_j)}$. Via SVD, obtain $\tilde{\bU}_{\bD} \in \RR^{m \times (p+1) K}$, the first $(p+1) K$ left singular vectors of $(\tilde{\bU}_{\bB},\hat{\bU}_{\bC_1},...,\hat{\bU}_{\bC_p})$. Then estimate the projection matrix $\tilde{\bP}_{\bD} = \tilde{\bU}_{\bD} \tilde{\bU}_{\bD}^T$. 
 
 5. Obtain $\hat{\bTheta}$ by solving (\ref{eq:step 5}) with $\tilde{\bP}_{\bD}$ in lieu of $\hat{\bP}_{\bD}$.

\end{algorithm}

\section{Simulation Results}\label{sec:simulation}

We compare our methods outlined in Algorithms \ref{algo:homoscedastic} and \ref{algo:heteroscedastic} with the following list of competitors: 
\begin{itemize}
    \item Oracle: the estimator obtained from (\ref{eq:step 5}) by using the true values of $\bB$ and $\bC$'s to construct $\bP_D$. This estimator is non-practical and only served as a benchmark. 
    \item Non-interaction: the estimator obtained assuming there is no interaction between $\bX$ and $\bZ$ in model (\ref{eq:first_model}). The estimation procedures are adapted from \cite{bing2022adaptive}. We similarly consider two versions of the algorithms with the homoscedastic and heteroscedastic noise.
    \item OLS: the ordinary least squares estimator without accounting for the presence of hidden variables.
\end{itemize}

To simplify the presentation, we mainly focus on the non-interaction method for comparison against our method, since the recently proposed methods in surrogate variable analysis (SVA)  that seek to account for hidden variables are similar to the non-interaction method presented above. 

\subsection{Data generating mechanism}\label{subsec: data generating mechanism}

The design matrix is generated as $\bX_i \sim \mathcal{N}_p(0, \bSigma)$ independently for all $1 \leq i \leq n$ and for $\Sigma_{jk} = (-1)^{j+k} (0.5)^{|j-k|}$ for all $1\leq j,k \leq p$. Let $\bZ_i=\bpsi^T\bX_i+\bW_i$ with $\bpsi_{jk} \sim \eta \cdot \mathcal{N}(0.5,0.1)$ independently for all $1\leq j \leq p$ and $1 \leq k \leq K$, where $\eta$ controls the level of dependence the observed and hidden variables. Set $\bA_{\ell j}^T \sim \mathcal{N}(0.5,0.1)$, and $\bB_{\ell k}^T$ and $\bC_{j, \ell k}^T$ $\sim \mathcal{N}(0.1,1)$ independently for all $1 \leq \ell \leq m$, $1 \leq k \leq K$, and $1\leq j \leq p$. The stochastic error $\bW_i$ is sampled as $\bW_{ik} \sim \mathcal{N}(0,1)$ independently for all $1 \leq i \leq n$ and $1 \leq k \leq K$. In the homoscedastic setting, $\bE_{i\ell} \sim \mathcal{N}(0,1)$ independently for all $1 \leq i \leq n$ and $1 \leq \ell \leq m$; and for the heteroscedastic setting, $\bE_{i \ell} \sim \mathcal{N}(0,\tau_{\ell}^2)$ where $\tau_{\ell}^2 = m v_{\ell}^{\alpha} / \sum_{\ell} v_{\ell}^{\alpha} \cdot (p+1)$ and $v_{\ell} \sim \textrm{Unif}\,[0,1]$ independently for $1 \leq \ell \leq m$. This specification of $\tau_{\ell}$ guarantees the overall level of variability scales with the number of hidden interaction terms (i.e. $\sum_{\ell} \tau_{\ell}^2 = p+1$) with larger values of $\alpha$ corresponding to a higher degree of heteroscedasticity.   

The number of hidden variables $K$ is fixed at 3 and is assumed to be known in all methods except in the experiments concerning selection of $K$ in Section \ref{sec:k-selection}.

\subsection{Experiment Setup}

We fix $p=2$ and consider 2 separate settings: (i) small $m$ and large $n$ ($m=25$, $n=1000$) and (ii) large $m$ and small $n$ ($m=500, n=100$). For the homoscedastic case, we iterate over $\eta \in \{0.1,0.3,...,1.1,1.3\}$; while, for the heteroscedastic case, we iterate over $\alpha \in \{0,3,6,...,12,15\}$ and fix $\eta = 0.5$. Within each parameter setting, we independently generate 100 datasets according to Section \ref{subsec: data generating mechanism} and report the average Sum Squared Error (SSE) in log scale $\log (\frac{1}{m} \|\hat{\bTheta} - \bA\|_F^2)$, 
and the average Prediction Mean Squared Error (PMSE) in log scale $\log (\frac{1}{n^*m} \|\bY^* - \hat{\bTheta}^T \bX^* \|_F^2)$ on a newly generated test set $(\bX^*, \bY^*)$ with $n^* = 5000$ data points. 

\subsection{Results and Discussion}

\subsubsection{SSE}\label{sec:sse}

\begin{figure}[h]
\caption{SSE $\log (\frac{1}{m} \|\hat\bTheta - \bA\|_F^2)$ for homoscedastic (top row) and heteroscedastic (bottom row) settings.}
\centering
\includegraphics[width=0.90\textwidth]{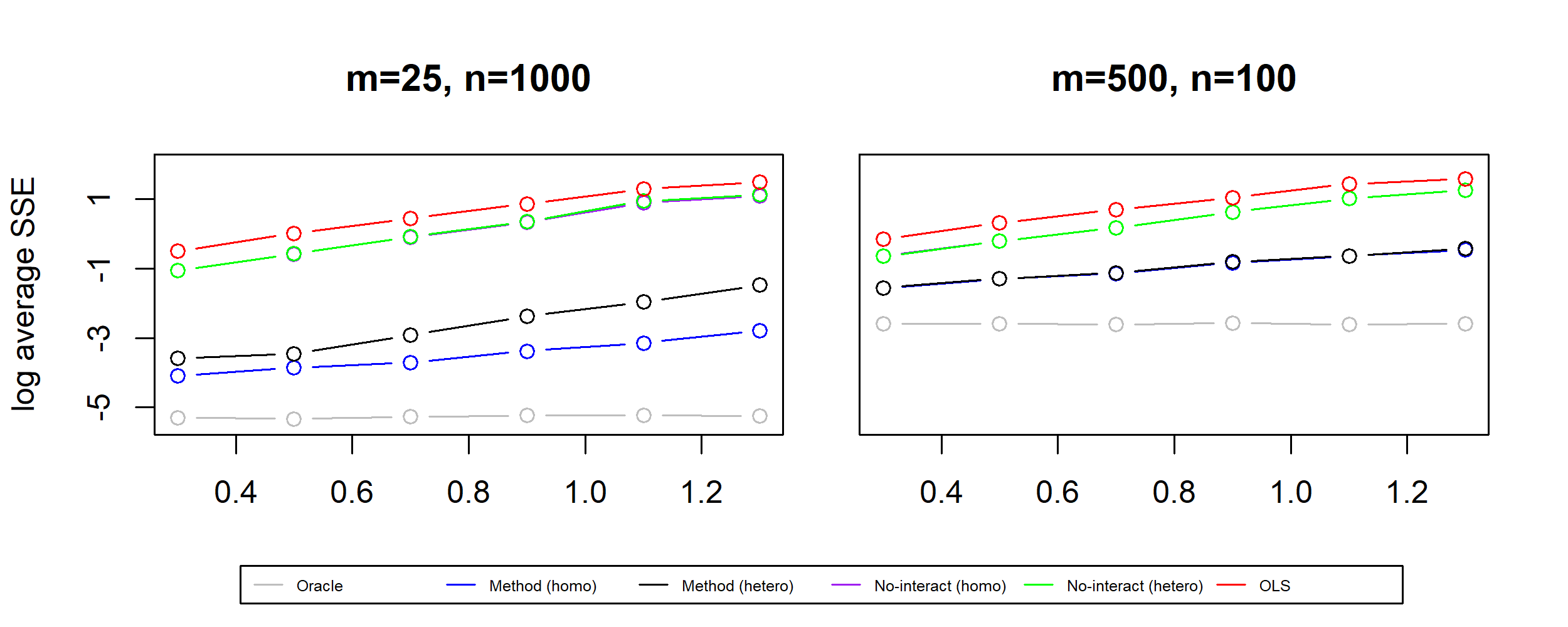}
\includegraphics[width=0.90\textwidth]{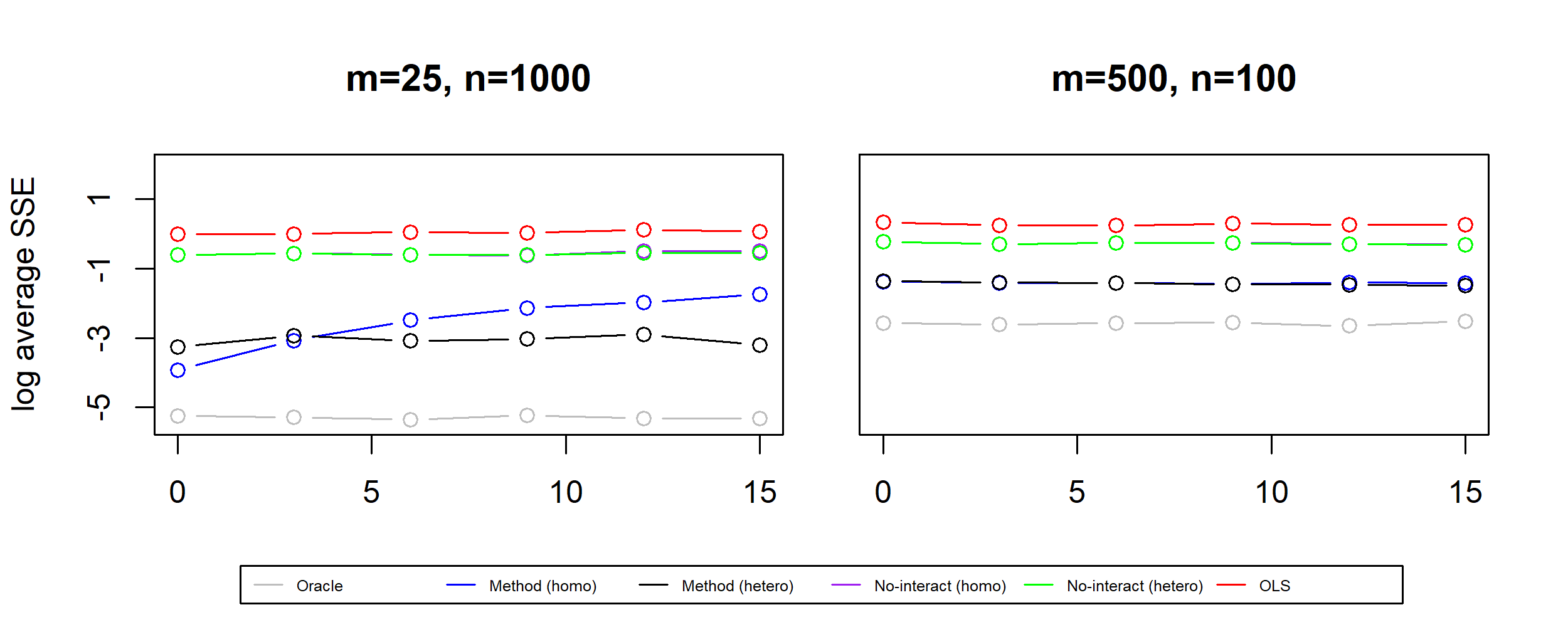}
\label{fig:SSE}
\end{figure}

\textit{Homoscedastic setting.} As seen in the top row of Figure \ref{fig:SSE}, our method outperforms the other methods across both settings whether or not we use HeteroPCA. In particular, it outperforms the methods without considering the interactions among the treatment  and the hidden variables, which shows the importance of accounting for the effect of interactions in this model. Under setting 1, where we have sufficiently large amount of data relative to the number of responses ($n\gg m$), the proposed Algorithm \ref{algo:homoscedastic} tailored for the homoscedastic noise outperforms our Algorithm \ref{algo:heteroscedastic}. Under setting 2, when $m$ is sufficiently large, the two algorithms perform relatively the same. 

\textit{Heteroscedastic setting.} As illustrated in the bottom row of Figure \ref{fig:SSE}, our method once again outperforms the other methods across both settings. On the other hand, employing Algorithm \ref{algo:heteroscedastic} with HeteroPCA for setting 1 yields substantially better results over using Algorithm \ref{algo:homoscedastic} with PCA. This suggests that our algorithm tailored for the heteroscedastic setting, when the underlying noise has a high degree of heteroscedasticity, is most preferable for large $n$ and small $m$. This is consistent with the discussion following Theorem \ref{thm:thm 2}. Similar to the homoscedastic setting, in setting 2 when $m$ is large enough, the two algorithms yield very similar results.

\subsubsection{PMSE}\label{sec:testmses}

\begin{figure}[h]
\caption{PMSE $\log (\frac{1}{n^*m} \|\bY^* - \hat{\bTheta}^T \bX^* \|_F^2)$ for homoscedastic settings.}
\centering
\includegraphics[width=0.90\textwidth]{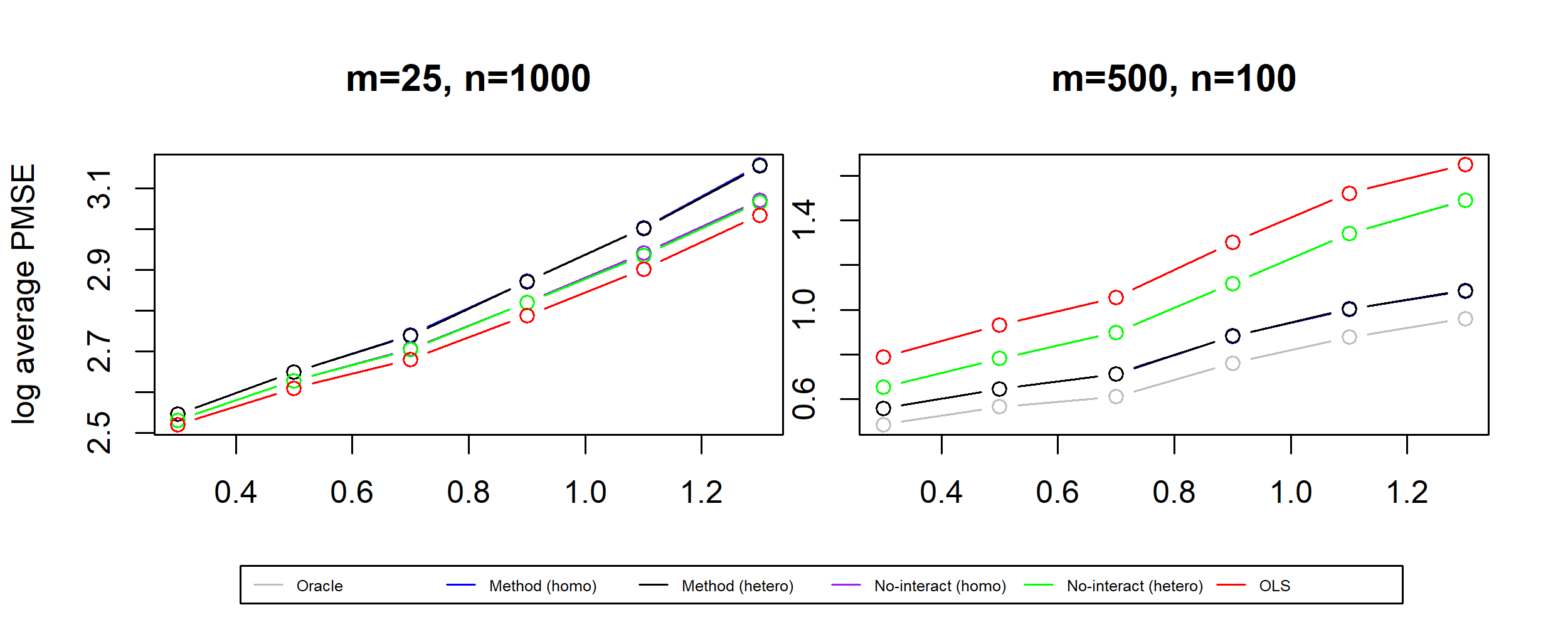}
\label{fig:PMSE}
\end{figure}


The PMSE as depicted in Fig \ref{fig:PMSE} exhibits a more interesting behavior than SSE. It appears that the prediction performance of our method is heavily dependent on the relative magnitude of $m$ and $n$. When $m$ is small and $n$ is large, the prediction error (in the test set) of our method is similar to the rest of the methods. One possible explanation is that, while our method can remove the bias due to the hidden variables, this effect is dominated by the variance of the noise $\Var(\bE)$ in the test prediction error, so that all the methods including OLS yield similar results. However, in the large $m$ and small $n$ regime, our interaction based methods lead to much smaller test error, as correcting for the bias due to the hidden variables becomes more imperative. While we mainly focus on the treatment effect estimation instead of prediction in this paper, our numerical results also demonstrate the favorable performance of our method in prediction when $m$ is relatively large. The PMSE under the heteroscedastic setting has a similar pattern and is omitted.




\subsubsection{Selection of $K$}\label{sec:k-selection} 

In practice, the number of hidden variables $K$ is often unknown. In this section, we propose a practical approach to estimate $K$ via a variant of the ratio test proposed by \cite{bing2022adaptive}. In particular, we select $K$ as the most common index of the largest eigenvalue gap across the coefficient matrices $\hat{\bphi}_B$ and $\hat{\bphi}_{C_j}$ obtained from step 3 of our algorithms. More formally,  define
\begin{align*}
    \hat{K} = \arg \max_{i\in\{1,...,K^*\}} | S_i |, ~~\textrm{where}~~S_i=\Big\{ j \in \{1,...,p+1\} :  \frac{\hat{\lambda}_{j,i}}{\hat{\lambda}_{j,i+1}} \geq \frac{\hat{\lambda}_{j,k}}{\hat{\lambda}_{j,k+1}} \forall k \neq i\Big\},
\end{align*}
$\hat{\lambda}_{1,i}$ and $\hat{\lambda}_{j,i}$, $1\leq i \leq m$, are the ordered nonincreasing eigenvalues of $\hat{\bphi}_B$ and $\hat{\bphi}_{C_j}$, respectively, and $K^*$ is an upper bound, often $\frac{\lfloor (n \wedge m) \rfloor}{2}$, on the number of hidden variables. In other words, the set $S_i$ includes all the indices $j$ such that the eigenvalue ratio $\frac{\hat{\lambda}_{j,i}}{\hat{\lambda}_{j,i+1}}$ is maximized at $i$. Then $\hat K$ corresponds to the one with the largest cardinality of $S_i$. This approach can be viewed as a generalization of the elbow method in clustering and PCA.

To evaluate the performance of this approach, we consider the following experiments. Recall that $\lambda_{K}(\bB^T \bSigma_W \bB)$ is the $K$th largest eigenvalue of $\bB^T \bSigma_W \bB$. We define the signal-to-noise ratio (SNR) as 
\begin{align*}
    \frac{1}{m} \lambda_{K}(\bB^T \bSigma_W \bB),
\end{align*}
which quantifies how the $K$th largest eigenvalue is separated from 0. In our experiments, we generate $\bB_{\ell k}^T \sim \mathcal{N}(0.1,1)$ independently for all $1 \leq \ell \leq m$ and $1 \leq k \leq K$, and vary the SNR by setting $\bSigma_W = \sigma_W^2 \bI_K$ across $\sigma_W \in \{0.1,0.3,...,1.3,1.5\}$ while  fixing $\alpha = 0$ and $\eta = 0.5$. For comparison, we also consider the method for selecting $K$ used in the Non-interaction approach \citep{bing2022adaptive}. We also consider two settings: (i) small $m$ and large $n$ ($m=25, n=1000$) and (ii) large $m$ and large $n$ ($m=500, n=1000$). 

From Figures \ref{fig:k_selection1} and \ref{fig:k_selection2}, we see that the non-interaction-based method consistently overestimates the number of hidden variables even when the SNR is large. On the other hand, our method consistently selects the correct value of $K$ (i.e., $K=3$) for large enough SNR. 

\begin{figure}[h]
\caption{$\hat{K}$ for non-interaction and interaction based methods $n \gg m$. The green dashed line indicates the true number of hidden features $K=3$.}
\centering
\includegraphics[width=0.90\textwidth]{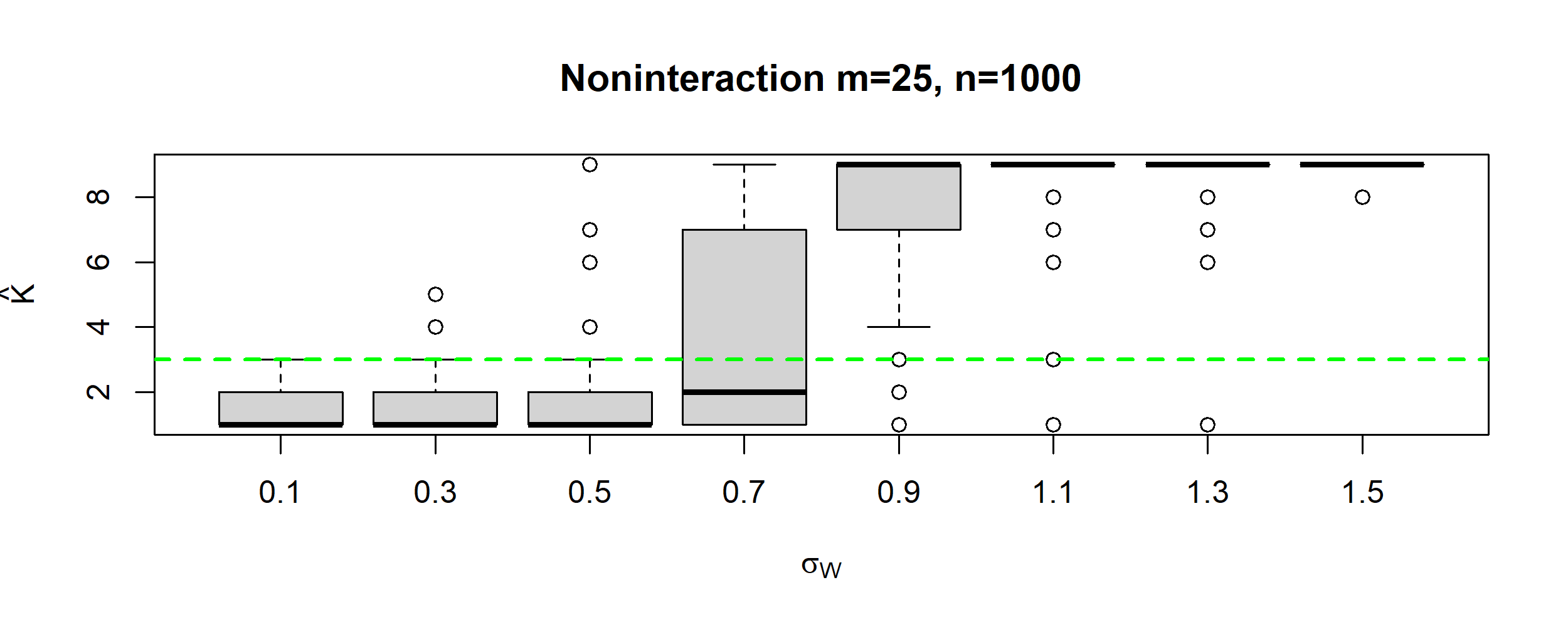}
\includegraphics[width=0.90\textwidth]{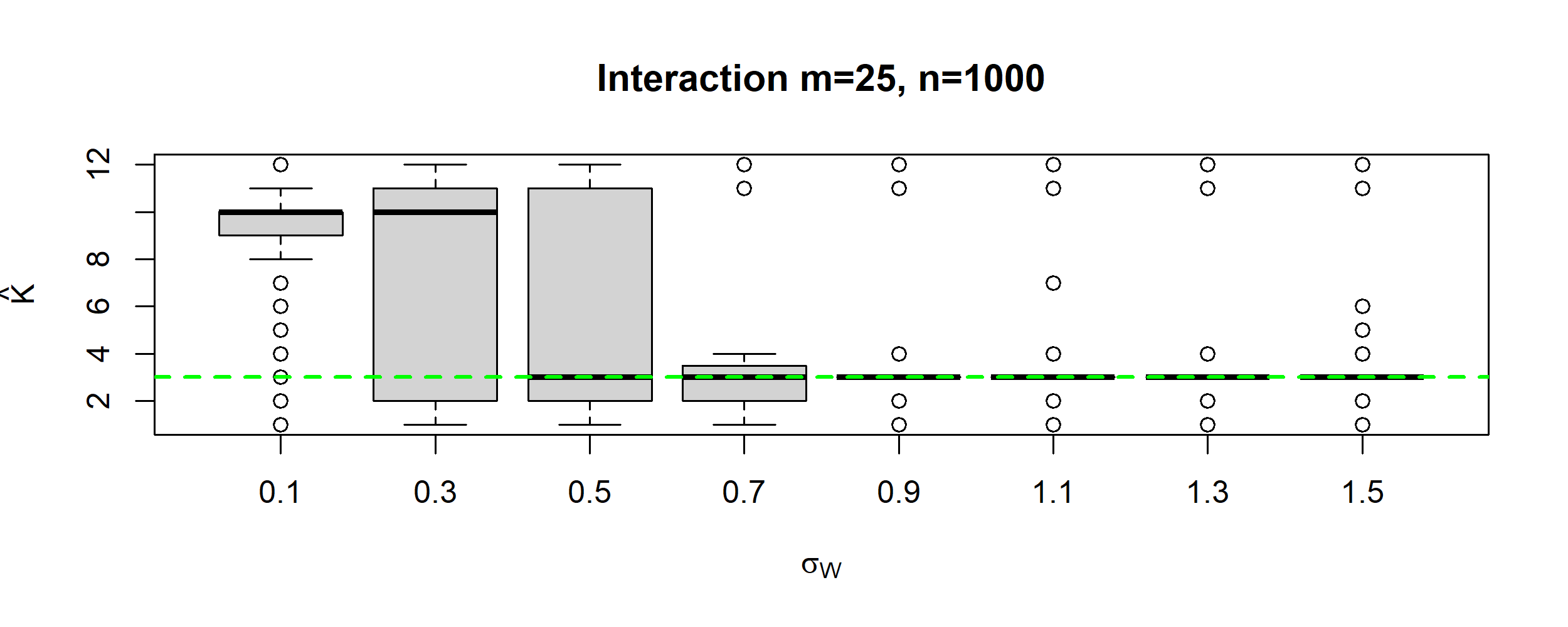}
\label{fig:k_selection1}
\end{figure}

\begin{figure}[h]
\caption{$\hat{K}$ for non-interaction and interaction based methods $m\gg n$. The green dashed line indicates the true number of hidden features $K=3$.}
\centering
\includegraphics[width=0.90\textwidth]{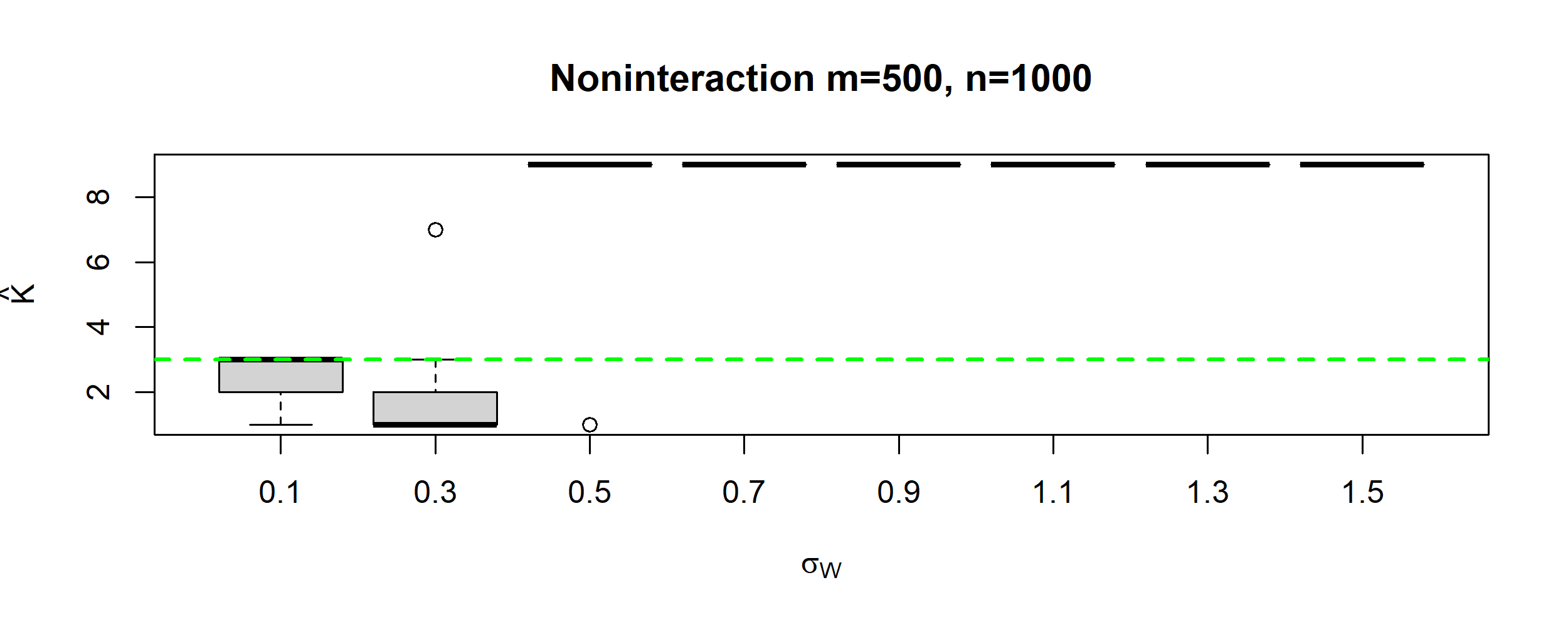}
\includegraphics[width=0.90\textwidth]{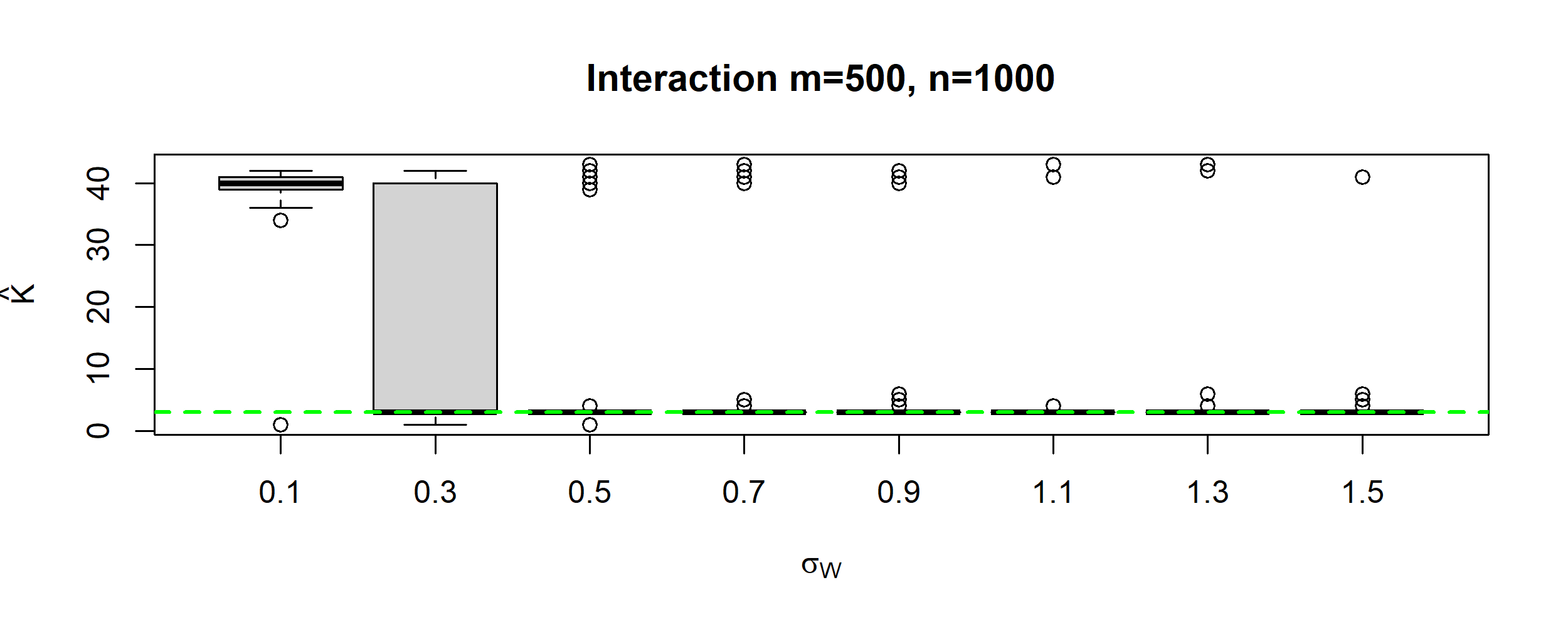}
\label{fig:k_selection2}
\end{figure}

\newpage

\section{Real Data}\label{sec:dataapplication}

We consider a microarray dataset gathered by \cite{vawter2004gender} which is comprised of postmortem microarray samples of 5 men and 5 women from 3 separate regions of the brain. This dataset was originally curated to answer questions surrounding gender differences in the prevalence of neuropsychiatric diseases and has since been used to test various statistical methods that seek to account for hidden variables (see \cite{gagnon2012using} and \cite{wang2017confounder}).

Each individual was sampled by 3 different universities, resulting in $10 \times 3 \times 3 = 90$ chips of which 6 were missing, leaving a total of $n=84$ samples. Of the remaining $84$ samples, $3$ were replicated to yield $87$ data points. In our analysis, we decide to omit the replicated samples. The sequences were preprocessed via Robust Multichip Average (\cite{irizarry2003exploration}) and then normalized feature-wise to yield readings of length $m=12600$. We take gender to be the primary variable with the brain pH level at the time of sampling to be a secondary covariate for a total of $p=2$ observed features.

In our data analysis, since the true parameter value is unknown, we decided to compute the test prediction MSE (defined in Section \ref{sec:testmses}) via 10-fold cross validation for the same 5 methods as the simulation study (we exclude the oracle, since it is unavailable). To select the number of hidden variables, we chose $\hat{K}$ as in Section \ref{sec:k-selection}. As seen in Table \ref{tab:real_data}, our interaction based method for the homoscedastic setting (i.e., Algorithm \ref{algo:homoscedastic}) achieves the lowest cross-validated prediction error. In Section \ref{sec:testmses}, our simulations suggested that one should not ignore the interaction present between observed and unobserved covariates for data with large $m$ and small $n$ when attempting to predict the response $\bY$. Our results here corroborate this claim.

\begin{table}
    \centering
    \begin{tabular}{|c|c|}
        \hline
         Method & cross validated PMSE  \\
         \hline
         OLS & 1.0602 \\
         \hline
         Non-interaction (homo) & 1.0493 \\ 
         Non-interaction (hetero) & 1.0499 \\
         \hline 
         Interaction (homo) & \textbf{1.0252} \\
         Interaction (hetero) & 1.0308 \\
         \hline
    \end{tabular}
    \caption{Cross validated PMSE for various methods on real data.}
    \label{tab:real_data}
\end{table}

\section*{Acknowledgment}
Ning is supported in part by National Science Foundation (NSF) CAREER award DMS-1941945
and NSF award DMS-1854637.

\bibliographystyle{abbrvnat}
\bibliography{ref}

\appendix
\section{Proof of Theorem \ref{thm:thm 1}}

\begin{proof} We prove the result under $p=1$ since the proof can be easily generalized to the case when $p>1$. Throughout the proof, denote $\bC_{1} = \bC \in \RR^{m \times K}$, $\bL_{2,11}^T = \bL_2^T \in \RR^m$, $\bL^T = [\bL_1^T, \bL_2^T] \in \RR^{m \times 2}$ and $\hat{\bL}^T = [\hat{\bL}_1^T, \hat{\bL}_2^T]$. 

For step 1, we may follow the standard results for ordinary least square estimators using Assumption \ref{ass_2} and show that $\|\hat{\bL}-\bL\|_F^2=\mathcal{O}_p(m/n)$. 

In step 3, for all $(r,s) \in [m] \times [m]$, we regress $\bY_{rs} := ((\hat{\bepsilon}_i \hat{\bepsilon}_i^T)_{rs})_{i \in [n]}$ on $\bX := ([1, X_i, X_i^2])_{i\in[n]}$ to obtain $[(\hat{\bphi}_{\bB})_{rs}, \hat{(\bphi}_{(\bB, \bC)})_{rs},  (\hat{\bphi}_{\bC})_{rs}]^T$ which, for brevity, we denote as $\hat{\bbeta}_{rs}$. Let $\tilde\bbeta_{rs}$ be the least square estimator using the response $((\bepsilon_i \bepsilon_i^T)_{rs})_{i \in [n]}$, and denote $\bbeta_{rs}$ as the true value of $\tilde\bbeta_{rs}$, that is the parameter value that minimizes $\EE(\bphi_s\bphi_r)$. Then,
\begin{align*}
    \|\sqrt{n} (\hat{\bbeta}_{rs} - \bbeta_{rs})\|_2^2 \leq \|\sqrt{n} (\tilde{\bbeta}_{rs} - \bbeta_{rs})\|_2^2 + \|\sqrt{n} (\hat{\bbeta}_{rs} - \tilde{\bbeta}_{rs})\|_2^2 \leq \mathcal{O}_p(1) + \mathcal{O}_p(1)
\end{align*}
where the former term follows from the standard results on least square estimator and Assumption \ref{ass_2}. The latter term deserves some explanation which we now provide. For ease of notation, denote $((\hat{\bepsilon}_i \hat{\bepsilon}_i^T)_{rs})_{i \in [n]}$ as $\hat{\bepsilon} \hat{\bepsilon}^T$.
\begin{align*}
    \|\sqrt{n}(\tilde{\bbeta}_{rs} - \hat{\bbeta}_{rs})\|_2^2 &= \|\sqrt{n} (\bX^T \bX)^{-1} \bX^T (\hat{\bepsilon} \hat{\bepsilon}^T -\bepsilon \bepsilon^T)\|_2^2 \\
    &\leq \|(\bX^T \bX)^{-1}\|_F^2 \cdot \|\sqrt{n} \bX^T(\hat{\bepsilon} \hat{\bepsilon}^T -\bepsilon \bepsilon^T)\|_2^2 \\
    &= \mathcal{O}_{p}(n^{-2}) \cdot \mathcal{O}_{p}(n^2) = \mathcal{O}_p(1).
\end{align*}
Adding 0 and using the triangle inequality on the latter term yields:
\begin{align*}
    \|\sqrt{n} \bX^T(\hat{\bepsilon} \hat{\bepsilon}^T -\bepsilon \bepsilon^T)\|_2^2 
    &= \|\sqrt{n} \bX^T(\hat{\bepsilon} \hat{\bepsilon}^T - \bepsilon \hat{\bepsilon}^T + \bepsilon \hat{\bepsilon}^T - \bepsilon \bepsilon^T)\|_2^2 \\
    &\leq 2 \, \|\sqrt{n} \bX^T (\hat{\bepsilon} - \bepsilon) \hat{\bepsilon}^T\|_2^2
    + 2 \, \|\sqrt{n} \bX^T \bepsilon (\hat{\bepsilon}^T - \bepsilon^T)\|_2^2,
\end{align*}
where $(\hat{\bepsilon} - \bepsilon) \hat{\bepsilon}^T$, for example, is understood to mean $((\hat{\epsilon}_i - \epsilon_i)\hat{\epsilon}_i^T)_{rs, 1\leq i \leq n}$. Consider the first term, $\sqrt{n} \bX^T (\hat{\bepsilon} - \bepsilon) \hat{\bepsilon}^T$. Its $j$th entry, for $j\in\{1,2,3\}$, is given by:
\begin{align*}
&\sqrt{n} \sum_{i=1}^{n} X_i^{j-1} \left( (L_{r1} - \hat{L}_{r1}) X_i + (L_{r2} - \hat{L}_{r2}) X_i^2 \right) \cdot ( Y_{is} - \hat{L}_{s1} X_i - \hat{L}_{s2} X_i^2 ) \\
    &= \sqrt{n} \sum_{i=1}^{n} \bigl( (L_{r1} - \hat{L}_{r1}) X_i^j Y_{is} + (L_{r2} - \hat{L}_{r2}) X_i^{j+1} Y_{is} \\
    &~~~- \hat{L}_{s1} (L_{r1} - \hat{L}_{r1}) X_i^{j+1} - \hat{L}_{s1} (L_{r2} - \hat{L}_{r2}) X_i^{j+2} \\
    &~~~- \hat{L}_{s2} (L_{r1} - \hat{L}_{r1}) X_i^{j+2} - \hat{L}_{s2} (L_{r2} - \hat{L}_{r2}) X_i^{j+3} \bigr).
\end{align*}
Squaring the above and noting the first and second terms dominate the convergence rate, we obtain: 
\begin{align*}
    n \cdot (L_{r1} - \hat{L}_{r1})^2 \cdot (\sum_{i=1}^{n} X_i^j Y_{ij})^2 = \mathcal{O}_{p}(n) \cdot \mathcal{O}_{p}(n^{-1}) \cdot \mathcal{O}_{p}(n^2) = \mathcal{O}_{p}(n^2),
\end{align*}
where we used the convergence rate from step 1 for the 2nd term and the law of large numbers for the 3rd term. The other terms are handled similarly, so we conclude that $\|\hat{\beta}_{rs} - \beta_{rs}\|_2^2 = \mathcal{O}_{p}(1/n)$. Following a similar analysis, we can show that 
\begin{align*}
    \| \sqrt{n} (\hat{\bphi}_{\bB} - \bphi_{\bB}) \|_F^2 = \mathcal{O}_p(m^2); \quad \quad \| \sqrt{n} (\hat{\bphi}_{\bC} - \bphi_{\bC}) \|_F^2 = \mathcal{O}_p(m^2).
\end{align*}

For step 4, write the eigendecomposition of the coefficient matrices from step 3 as 
\begin{align*}
    \bphi_B = \bU_B \bD_B \bU_B^T; \quad \quad \bphi_C = \bU_C \bD_C \bU_C^T,
\end{align*}
the eigendecomposition of the estimated coefficient matrices as
\begin{align*}
    \hat{\bphi}_B = \hat{\bU}_B \hat{\bD}_B \hat{\bU}_B^T; \quad \quad \hat{\bphi}_C = \hat{\bU}_C \hat{\bD}_C \hat{\bU}_C^T,
\end{align*}
and the SVD of the concatenated matrices of the 2 sets of $K$ leading eigenvectors as
\begin{align*}
    \begin{bmatrix} \bU_B & \bU_C \end{bmatrix} = \bU \bD \bV^T ; \quad \quad \begin{bmatrix} \hat{\bU}_B & \hat{\bU}_C \end{bmatrix} = \hat{\bU} \hat{\bD} \hat{\bV}^T. 
\end{align*}
Then, $\bU$ are the first $2K$ eigenvectors of $\bSigma := \begin{bmatrix} \bU_B & \bU_C \end{bmatrix} \begin{bmatrix} \bU_B^T \\ \bU_C^T \end{bmatrix}$. Similarly, $\hat{\bU}$ are the first $2K$ eigenvectors of $\hat{\bSigma} := \begin{bmatrix} \hat{\bU}_B & \hat{\bU}_C \end{bmatrix} \begin{bmatrix} \hat{\bU}_B^T \\ \hat{\bU}_C^T \end{bmatrix}$. Denote $\lambda_B$ and $\lambda_C$ as the $K$th largest eigenvalue of $\bB^T \bSigma_W \bB$ and $\bC^T \bSigma_W \bC$, respectively. Now, by a variant of the Davis-Kahan theorem (Theorem 2 in \citep{yu2015useful}), with $r=1, s=2K, d:=s-r+1=2K$, and the convergence rate obtained in step 3, $\exists$ orthogonal matrices $\bO_B, \bO_C \in \RR^{K \times K}$ s.t. 
\begin{align}\label{eq:DK phi_B}
    \| \hat{\bU}_B \bO_B - \bU_B\|_F \leq 2^{5/2} \, \|\hat{\bphi}_B - \bphi_B\|_F \cdot \left( \lambda_K(\bB^T \bSigma_W \bB) \right)^{-1} = \mathcal{O}_p(\frac{m}{\lambda_B\sqrt{n}}),
\end{align}
\begin{align}\label{eq:DK phi_C}
    \| \hat{\bU}_C \bO_C - \bU_C\|_F \leq 2^{5/2} \, \|\hat{\bphi}_C - \bphi_C\|_F \cdot \left( \lambda_K(\bC^T \bSigma_W \bC) \right)^{-1} = \mathcal{O}_p(\frac{m}{\lambda_C\sqrt{n}}),
\end{align}
where we used the fact that $\bSigma_E = \sigma^2 \bI_m$ and $\text{rank}(\bSigma_W) = K$, so that $\lambda_K(\bphi_B) - \lambda_{K+1}(\bphi_B) = \lambda_K(\bB^T \bSigma_W \bB)$. 

Since $\bP_{\bD} = \bU\bU^T$ and $\hat{\bP}_{D} = \hat{\bU} \hat{\bU}^T$, it now follows that, by another application of the Davis-Kahan Theorem, $\exists$ orthogonal matrix $\bO \in \RR^{2K \times 2K}$ s.t.
\begin{align}
    \|\hat{\bP}_D - \bP_D\|_F &= \|{\hat{\bU}}{\hat{\bU}}^T - {\bU} {\bU}^T\|_F \notag \\
    &= \|{\hat{\bU}} {\hat{\bU}}^T - {\bU} \bO^T {\hat{\bU}}^T + {\bU} \bO^T {\hat{\bU}}^T - {\bU}{\bU}^T \|_F \notag \\
    &\leq \|({\hat{\bU}} \bO - {\bU})\bO^T {\hat{\bU}}^T\|_F + \|{\bU} ({\hat{\bU}} \bO - {\bU})^T\|_F \notag \\
    &\leq 2 \, \|{\hat{\bU}} \bO - {\bU}\|_F \notag \\
    &\leq 2 \cdot 2^{3/2} \,  \| {\hat{\bSigma}} - {\bSigma}\|_F \cdot \left( \lambda_{2K}({\bSigma}) - \lambda_{2K+1}({\bSigma}) \right)^{-1} \notag \\
    &= 2^{5/2} \, \| {\hat{\bSigma}} - {\bSigma} \|_F \cdot \left( \lambda_{2K}({\bSigma}) \right)^{-1}. \label{eq:DK P_D}
\end{align}
Since the matrix $\begin{bmatrix} \bU_B & \bU_C \end{bmatrix}$ has rank $2K$ by Assumption \ref{ass_3}, $\bU_B$ and $\bU_C$ cannot be linearly dependent. Hence, for any column $v\in\RR^m$ of $\bU_B$ and $\bU_C$ (there are $2K$ of them), we have $v^T\bSigma v=v^T(\bU_B \bU_B^T + \bU_C \bU_C^T) v\geq 1$, which implies $\lambda_{2K}({\bSigma})\geq 1$. 

Finally, it suffices to upper bound $\| {\hat{\bSigma}} - {\bSigma}\|_F$. Similar to derivation of (\ref{eq:DK P_D}), we have
\begin{align*}
\| {\hat{\bSigma}} - {\bSigma}\|_F&\leq \| \bU_B \bU_B^T - \hat{\bU}_B \hat{\bU}_B^T\|_F+\| \bU_C \bU_C^T - \hat{\bU}_C \hat{\bU}_C^T\|_F=\mathcal{O}_p(\frac{m}{\lambda_B\sqrt{n}}+\frac{m}{\lambda_C\sqrt{n}}),
\end{align*}
where we use (\ref{eq:DK phi_B}) and (\ref{eq:DK phi_C}) in the last step. 

Thus, we see that 
\begin{align}\label{eq:rate P_D}
    \|\hat{\bP}_{\bD} - \bP_{\bD}\|_F = (\mathcal{O}_p(\frac{m}{\sqrt{n}}) \cdot \lambda_B^{-1} + \mathcal{O}_p(\frac{m}{\sqrt{n}}) \cdot \lambda_C^{-1}) \cdot \left(\lambda_{2K}(\bSigma)\right)^{-1}.
\end{align}
Under Assumption \ref{ass_3} and the earlier remark showing $\lambda_{2K}(\Sigma) \geq 1$, the convergence rate in (\ref{eq:rate P_D}) can be further simplified to
\begin{equation}\label{eq:rate simplified P_D}
    \|\hat{\bP}_{\bD} - \bP_{\bD}\|_F = \mathcal{O}_p(\frac{1}{\sqrt{n}}).
\end{equation}

For step 5, recall that $\bTheta^T=\bP_{\bD}^\perp\bA^T$ and $\bX$ is the data matrix in $\RR^{n\times p}$. Adding and subtracting the least squares estimator with $\bP_{D}$ known and using the triangle inequality, we obtain:
\begin{align*}
    \|\hat{\bTheta} - \bTheta\|_F^2 = \sum_{j=1}^{m} \|\hat{\bTheta}_{(j,\cdot)} - \bTheta_{(j,\cdot)}\|_2^2 &= \sum_{j=1}^{m} \|(\bX^T \bX)^{-1} \bX^T (\hat{\bP}^{\perp}_D \bY)_{(j,\cdot)}^T - (\bP^{\perp}_D \bA^T)_{(j,\cdot)}^T\|_2^2 \\
    &\leq 2\sum_{j=1}^{m} \|(\bX^T \bX)^{-1} \bX^T (\hat{\bP}^{\perp}_D \bY)_{(j,\cdot)}^T - (\bX^T \bX)^{-1} \bX^T (\bP^{\perp}_D \bY)_{(j,\cdot)}^T\|_2^2 \\
    &~~+ 2\sum_{j=1}^{m} \|(\bX^T \bX)^{-1} \bX^T (\bP^{\perp}_D \bY)_{(j,\cdot)}^T - (\bP^{\perp}_D \bA^T)_{(j,\cdot)}^T\|_2^2.
\end{align*}
We bound each of these terms in turn. The first is bounded above as: 
\begin{align*}
    &\sum_{j=1}^{m} \|(\bX^T \bX)^{-1} \bX^T (\hat{\bP}^{\perp}_D \bY)_{(j,\cdot)}^T - (\bX^T \bX)^{-1} \bX^T (\bP^{\perp}_D \bY)_{(j,\cdot)}^T\|_2^2 \\
    &\leq \|(\bX^T \bX)^{-1} \bX^T \bY^T\|_2^2 \cdot \sum_{j=1}^{m} \|(\hat{\bP}_D - \bP_D)_{(j,\cdot)}\|_2^2 \\
    &= \|(\bX^T \bX)^{-1}\|^2_2 \cdot \|\bX^T \bY^T\|_F^2 \cdot \|\hat{\bP}_{D} - \bP_D\|_F^2 \\
    &= \mathcal{O}_p(\frac{1}{n^2}) \cdot \mathcal{O}_p(m n^2) \cdot \mathcal{O}_p(\frac{1}{n}) = \mathcal{O}_p(\frac{m}{n}).
\end{align*}
For the second term, we note that $\bP^{\perp}_D \bY=\bP^{\perp}_D \bA^T\bX^T+\bP^{\perp}_D \bE$. Then,
\begin{align*}
&\sum_{j=1}^{m} \|(\bX^T \bX)^{-1} \bX^T (\bP^{\perp}_D \bY)_{(j,\cdot)}^T - (\bP^{\perp}_D \bA^T)_{(j,\cdot)}^T\|_2^2\\
&=\|(\bX^T \bX)^{-1} \bX^T \bE^T\bP^{\perp}_D\|_F^2\\
&\leq \|\bP^{\perp}_D\|_2^2 \cdot \|(\bX^T \bX)^{-1} \bX^T \bE^T\|_F^2=\mathcal{O}_p(\frac{m}{n}).
\end{align*}
Altogether, we obtain $ \|\hat{\bTheta} - \bTheta\|_F^2=\mathcal{O}_p(\frac{m}{n})$. Finally, we note that $\bA^T-\bTheta^T=\bP_{\bD}\bA^T$. By Lemma \ref{Lemma 2}, we obtain
\begin{align*}
    \|\hat{\bTheta} - \bA\|_F^2 &\leq 2 \, \|\hat{\bTheta} - \bTheta\|_F^2 + 2 \, \|\bTheta - \bA\|_F^2=\mathcal{O}_p(\frac{m}{n}+K\eta).
\end{align*}
This completes the proof.
\end{proof}

\begin{lemma}\label{Lemma 2}
Under the same assumptions in Theorem \ref{thm:thm 2}, we have
    \begin{align*}
\frac{1}{m}\|\bP_D \bA^T\|_F^2=\mathcal{O}(K\eta/m),
    \end{align*}
where $\eta=\frac{1}{Km}\|\bD \bA^T\|_F^2$.
\end{lemma}

\begin{proof}
    Recall $\bD^T = \begin{bmatrix} \bB^T & \bC_1^T&...&\bC^T_p \end{bmatrix} \in \RR^{m\times (p+1) K}$. Then 
    \begin{align*}
        \|\bP_D \bA^T\|_F^2 &= \|\bD^T (\bD\bD^T)^{-1} \bD \bA^T\|_F^2 \\
        &\leq \|\bD^T (\bD\bD^T)^{-1}\|_{2}^2 \cdot \|\bD \bA^T\|_F^2 =\left(\lambda_{\min}(\bD\bD^T)\right)^{-1}  \cdot \|\bD \bA^T\|_F^2.
    \end{align*}
    Since $\bD$ has rank $(p+1)K$, the smallest eigenvalue of $\bD\bD^T$ is equal to $\lambda_D$, the $(p+1)K$th eigenvalue of $\bD^T\bD=\bB^T\bB+\sum_{j=1}^p\bC_j^T\bC_j$. By Assumption \ref{ass_3}, we know $\lambda_D\asymp m$. Thus, $\frac{1}{m}\|\bP_D \bA^T\|_F^2=\mathcal{O}(K\eta/m)$, where $\eta=\frac{1}{Km}\|\bD \bA^T\|_F^2$. 
\end{proof}

\section{Proof of Theorem \ref{thm:thm 2}}

\begin{proof}
    Like the proof of Theorem \ref{thm:thm 1}, we consider the case $p=1$. The only part of our estimation procedure that is changed under the heteroscedastic setting is the upper bound on the difference between our estimated projection matrix and the true projection matrix. In particular, 
    \begin{align}
        \|\tilde{\bP_D}-\bP_D\|_F \leq 2^{5/2} \, \|\tilde{\bSigma} - \bSigma\|_F \cdot \left(\lambda_{2K}(\bSigma)\right)^{-1} 
    \end{align}
    
    where 
    \begin{align*}
        \tilde{\bSigma} = \begin{bmatrix} \tilde{\bU}_B & \hat \bU_C \end{bmatrix} \begin{bmatrix} \tilde{\bU}_B^T \\ \hat \bU_C^T \end{bmatrix}; \quad \quad \bSigma = \begin{bmatrix} \bU_B & \bU_C \end{bmatrix} \begin{bmatrix} \bU_B^T \\ \bU_C^T \end{bmatrix}.
    \end{align*}
    
    Recall that $\bU_B$ are the first $K$ eigenvectors of $\bM=\bB^T \bSigma_W \bB$. Then $\hat{\bphi}_B$ can be decomposed into $\hat{\bphi}_B=\bM+\bZ_1+\bZ_2$, where $\bZ_1=\hat{\bphi}_B-\bphi_B$ with $\bphi_B=\bB^T \bSigma_W \bB+\Cov(\bE)$ and $\bZ_2=\Cov(\bE)$. The only term we need to consider in this different setting is 
    \begin{align}
        \|\tilde{\bU}_B \tilde{\bU}_B^T - \bU_B \bU_B^T\|_F &\leq 2 \, \|\sin \Theta(\tilde{\bU}_B, {\bU}_B)\|_F \notag\\
        &~\lesssim \frac{\|\Gamma(\bZ_2)\|_F+\|\bZ_1\|_F}{\lambda_B} \wedge \sqrt{K} \notag \\
        &~\lesssim (\frac{m}{\lambda_B \sqrt{n}}+\frac{\|\Gamma(\Cov(\bE))\|_F}{\lambda_B}) \wedge \sqrt{K}  \notag \\
        &~= \mathcal{O}_p(\frac{1}{\sqrt{n}}+\frac{\|\Gamma(\Cov(\bE))\|_F}{m}), \label{eq:DK P_D hetero}
    \end{align}
    where we used Theorem 13 in \cite{bing2022adaptive} with Assumption \ref{ass_4} to bound $\|\sin \Theta(\tilde{\bU}_B, {\bU}_B)\|_F$ and the same argument as in (\ref{eq:rate P_D}) of the proof of Theorem \ref{thm:thm 1} to upper bound $\|\bZ_1\|_F$. The rest of the proof follows the same line as the proof of Theorem \ref{thm:thm 1}. We omit the details. 
\end{proof}

\end{document}